\documentclass[english,11pt]{article}
\usepackage{fullpage}
\usepackage{amsmath}
\usepackage{amssymb}
\usepackage{babel}
\usepackage{graphics}
\usepackage{graphicx}
\usepackage{subfigure}
\usepackage{scalefnt}
\usepackage{rotating}

\usepackage{tikz}
\usetikzlibrary{arrows,automata,shapes,decorations,calc,matrix}
\newtheorem{theorem}{Theorem}

\newtheorem{lemma}[theorem]{Lemma}
\newtheorem{conjecture}[theorem]{Conjecture}

\newenvironment{proof}{\noindent {\bf Proof }}{\qed}

\newcommand{\qed}{\penalty 1000 \hfill\penalty 1000$\Box$\par\medskip}

\newcommand{\start}{{\rm start}}
\newcommand{\val}{{\rm val}}
\newcommand{\modS}{\!\!\mod}
\newcommand{\maxState}{x}

\makeatletter

\begin{document}
\title{
Strategy complexity of finite-horizon Markov decision processes and simple stochastic games\thanks{
Work of the second author supported by the Sino-Danish Center for the Theory of Interactive Computation,
funded by the Danish National Research Foundation and the National
Science Foundation of China (under the grant 61061130540). The second author acknowledge support from the Center for research in
the Foundations of Electronic Markets (CFEM), supported by the Danish
Strategic Research Council.
 The first author was supported by 
FWF Grant No P 23499-N23,  FWF NFN Grant No S11407-N23 (RiSE), ERC Start grant (279307: Graph Games), and Microsoft faculty fellows award.}}
\author{
Krishnendu Chatterjee\thanks{IST Austria. Email: {\tt krish.chat@ist.ac.at}.}
\and Rasmus Ibsen-Jensen\thanks{Department of Computer Science,
Aarhus University, Denmark. E-mail:
{\tt rij@cs.au.dk}.}
}
\date{}

\maketitle

\begin{abstract}
Markov decision processes (MDPs) and simple stochastic games (SSGs) provide a rich 
mathematical framework to study many important problems related to 
probabilistic systems. 
MDPs and SSGs with finite-horizon objectives, where the goal
is to maximize the probability to reach a target state in a given finite time, 
is a classical and well-studied problem.
In this work we consider the strategy complexity of finite-horizon MDPs and SSGs.
We show that for all $\epsilon>0$, the natural class of counter-based strategies 
require at most $\log \log (\frac{1}{\epsilon}) + n+1$ memory states, 
and memory of size $\Omega(\log \log (\frac{1}{\epsilon}) + n)$ is required,
for $\epsilon$-optimality, where $n$ is the number of states of the MDP (resp. SSG).
Thus our bounds are asymptotically optimal.
We then study the periodic property of optimal strategies, and show a 
sub-exponential lower bound on the period for optimal strategies.
\end{abstract}

\section{Introduction}

\smallskip\noindent{\bf Markov decision process and simple stochastic games.}
The class of \emph{Markov decision processes (MDPs)} is a classical model for probabilistic 
systems that exhibit both stochastic and and deterministic 
behavior~\cite{Howard}.
MDPs have been widely used to model and solve control problems for 
stochastic systems~\cite{FV97}: there, non-determinism represents the freedom 
of the controller to choose a control action, while the probabilistic 
component of the behavior describes the system response to control actions. 
\emph{Simple stochastic games (SSGs)} enrich MDPs by allowing two types of 
non-determinism (angelic and demonic non-determinism) along with stochastic 
behavior~\cite{Condon92}. 
MDPs and SSGs provide a rich mathematical framework to 
study many important problems related to probabilistic systems.

\smallskip\noindent{\bf Finite-horizon objective.} 
One classical problem widely studied for MDPs and SSGs is 
the \emph{finite-horizon objective}.
In a finite-horizon objective, a finite time horizon $T$ is given and the goal
of the player is to maximize the payoff within the time horizon $T$ in MDPs
(in SSGs against all strategies of the opponent). 
The complexity of MDPs and SSGs with finite-horizon objectives have been 
well studied, with book chapters dedicated to them~\cite{FV97,Puterman}.
The complexity results basically show that iterating the Bellman equation for
$T$ steps yield the desired result~\cite{FV97,Puterman}.
While the computational complexity have been well-studied, perhaps surprisingly
the strategy complexity has not received great attention.
In this work we consider several problems related to the strategy complexity 
of MDPs and SSGs with finite-horizon objectives, where the 
objective is to reach a target state within a finite time horizon $T$.

\smallskip\noindent{\bf Our contribution.} In this work we consider the 
memory requirement for $\epsilon$-optimal strategies, for $\epsilon>0$, 
and a periodic property of optimal strategies in finite-horizon MDPs and 
SSGs. 
A strategy is an $\epsilon$-optimal strategy, for $\epsilon>0$, if the 
strategy ensures within $\epsilon$ of the optimal value against all 
strategies of the opponent. 
For finite-horizon objectives, the natural class of strategies are counter-based
strategies, which has a counter to count the number of time steps.
Our first contribution is to establish asymptotically optimal memory bounds 
for $\epsilon$-optimal counter-based strategies, for $\epsilon>0$, in finite-horizon
MDPs and SSGs.
We show that $\epsilon$-optimal counter-based strategies require at 
most memory of size $\log \log (\frac{1}{\epsilon}) + n+1$ 
and memory of size $\Omega(\log \log (\frac{1}{\epsilon}) + n)$ is required,
where $n$ is the size of the state space.
Thus our bounds are asymptotically optimal. 
The upper bound holds for SSGs and the lower bound is for 
MDPs.
We then consider the periodic (or regularity) property of optimal strategies. 
The period of a strategy is the number $P$ such that the strategy repeats 
within every $P$ steps (i.e., it is periodic with time step $P$). 
We show a sub-exponential lower bound on the period of optimal strategies
for MDPs with finite-horizon objectives, by presenting a family of MDPs
with $n$ states where all optimal strategies are periodic and the period
is $2^{\Omega(\sqrt{n\cdot \log(n)})}$.

\smallskip\noindent{\bf Organization of the paper.}
The paper is organized as follows:
In Section~\ref{sec:def} we present all the relevant definitions related 
to stochastic games and strategies. 
In Section~\ref{sec:counter} we show that $\Theta(n+\log \log \epsilon^{-1})$ number of bits are 
necessary and sufficient for $\epsilon$-optimal counter-based strategies, for all $\epsilon>0$,
in both finite-horizon MDPs and SSGs.  
In Section~\ref{sec:period} we show that there are finite-horizon MDPs where all 
optimal strategies are periodic and have a period of $2^{\Omega(\sqrt{n\log n})}$.

\section{Definitions}\label{sec:def}
The class of {\em infinite-horizon simple stochastic games} (SSGs) consists of 
two player, zero-sum, turn-based games, played on a (multi-)graph. 
The class was first defined by Condon~\cite{Condon92}. 
Below we define SSGs, the finite-horizon version, and the important sub-class of MDPs.

\smallskip\noindent{\em SSGs, finite-horizon SSGs, and MDPs.}
An SSG $G=(S_1,S_2,S_R,\perp,(A_s)_{s\in S_1\cup S_2\cup S_R},s_0)$ consists of a terminal state $\perp$ and three sets of disjoint non-terminal states, $S_1$ (max state), $S_2$ (min states), $S_R$ (coin toss states). We will use $S$ to denote the union, i.e., $S=S_1\cup S_2 \cup S_R$. For each state $s\in S$, let $A_s$ be a (multi-)set of {\em outgoing arcs of $s$}. We will use $A=\bigcup_s A_s$ to denote the (multi-)set of all arcs. Each state $s\in S$ has two outgoing arcs. If $a$ is a arc, then $d(a)\in S\cup \{\perp\}$ is the {\em destination} of $a$. There is also a designated start state $s_0\in S$.   
The class of {\em finite-horizon simple stochastic games} (FSSGs) also consists
of two player, zero-sum, turn-based games, played on a (multi-)graph. 
An FSSG $(G,T)$ consists of an SSG $G$ and a finite time limit (or horizon) $T\geq 0$. 
Let $G$ be an SSG and $T\geq 0$, then we will write the FSSG $(G,T)$ as $G^T$.
Given an SSG $G$ (resp. FSSG $G^T$), for a state $s$, we denote by $G_s$ (resp. $G^T_s$) 
the same game as $G$ (resp. $G^T$), except that $s$ is the start state.
The class of {\em infinite (resp. finite) horizon Markov decision processes} (MDPs and FMDPs respectively) 
is the subclass of SSGs (resp. FSSGs) where $S_2=\emptyset$.

\smallskip\noindent{\em Plays and objectives of the players.}
An SSG $G$ is {\em played} as follows. A pebble is moved on to $s_0$. For $i\in \{1,2\}$, whenever the pebble is moved on to a state $s$ in $S_i$, then Player~$i$ chooses some arc $a\in A_s$ and moves the pebble to $d(a)$. Whenever the pebble is moved on to a state $s$ in $S_R$, then an $a\in A_s$ is chosen uniformly at random and the pebble moves to $d(a)$. If the pebble is moved on to $\perp$, then the game is over. 
For all $T\geq 0$ the FSSG $G^T$ is played like $G$, except that the pebble can be moved at most $T+1$ times.
The {\em objective} of both SSGs and FSSGs is for Player~1 to maximize the probability that the pebble is moved on to $\perp$ 
(eventually in SSGs and with in $T+1$ time steps in FSSGs). 
The objective of Player~2 is to minimize this probability.

\smallskip\noindent{\em Strategies.}
Let $S^*$ be the set of finite sequences of states. 
For all $T$, let $S^{\leq T}\subset S^*$ be the set of sequences of states, which have length at most $T$.
A {\em  strategy $\sigma_i$ for Player~$i$} in an SSG is a map from $S^*\times S_i$ into $A$, 
such that for all $w\in S^*$ and $s\in S$ we have $\sigma_i(w \cdot s)\in A_s$. 
Similarly, a {\em strategy $\sigma_i$ for Player~$i$} in an FSSG $G^T$ is a map from $S^{\leq T}\times S_i$ into $A$, 
such that for all $w\in S^{\leq T}$ and $s\in S$ we have $\sigma_i(w \cdot s)\in A_s$. 
In all cases we denote by $\Pi_i$ the set of all strategies for Player~$i$. If $S_i=\emptyset$, we will let $\emptyset$ denote the corresponding 
strategy set. 
Below we define some special classes of strategies. 

\smallskip\noindent{\em Memory-based, counter-based and Markov strategies.}
Let $M= \{0,1\}^*$ be the set of possible {\em memories}. 
A {\em memory-based strategy  $\sigma_i$ for Player~$i$} consists of a pair 
$(\sigma_u,\sigma_a)$, where 
\begin{itemize} 
\item{} $\sigma_u$, the memory-update function, is a map from $M\times S$ into $M$
\item{} $\sigma_a$, the next-action function, is a map from $M\times S_i$ into $A$, such that for 
all $m\in M$ and $s\in S_i$ we have $\sigma_a(m,s)\in A_s$.
\end{itemize}
A {\em counter-based strategy} is a special case of memory-based strategies, where for all $m\in M$ 
and $s,s'\in S$ we have $\sigma_u(m,s)=\sigma_u(m,s')$. That is the memory can only contain a counter of some type. 
We will therefore write $\sigma_u(m,s)$ as $\sigma_u(m)$ for all $m,s$ and any counter-based strategy $\sigma$.
A {\em  Markov strategy $\sigma_i$ for Player~$i$}  is a special case of strategies where 
\[\forall p,p'\in S^{\leq T}: |p|=|p'| \wedge p_{|p|}=p'_{|p'|}\in S_i \Rightarrow \sigma(p',p'_{|p'|})=\sigma(p,p_{|p|}).\] 
That is, a Markov strategy only depends on the length of the history and the current state. 
Let $\Pi'_i$ be the set of all Markov strategies for Player~$i$.

\smallskip\noindent{\em Following a strategy.}
For a  strategy, $\sigma_i$, for Player~$i$ we will say that Player~$i$ {\em follows} $\sigma_i$ if for all $n$ given the sequence of states $(p_i)_{i\leq n}$ the pebble has been on until move $n$ and that $p_n\in S_i$, then Player~$i$ chooses $\sigma((p_i)_{i\leq n},p_n)$. 
For a memory-based  strategy for Player~$i$ $\sigma_i$, we will say that Player~$i$ {\em follows} $\sigma_i$ if for all $n$ given the sequence of states $(p_i)_{i\leq n}$ the pebble has been on until move $n$, that $p_n\in S_i$ and that $m^i=\sigma_u(m^{i-1},p_i)$ and that $m^0=\emptyset$, then Player~$i$ chooses $\sigma_a(m^n,p_n)$. 

\smallskip\noindent{\em Space required by a memory-based strategy.}
The {\em space usages of a memory-based strategy} is the logarithm of the number of distinct states generated by the strategy at any point, if the player follows that strategy. A memory-based strategy is {\em memoryless} if there is only one memory used by the strategy. For any FSSG $G^T$ with $n$ states it is clear that the set of  strategies is a subset of memory-based strategies that uses memory at most $T\log n$, since for any  strategy $\sigma$ we can construct a memory-based strategy $\sigma'$ by using the memory for the sequence of states and then choose the same action as $\sigma$ would with that sequence of states. Hence we will also talk about $\epsilon$-optimal memory-based strategies. 
Also note that for any FSSG $G^T$ it is clear that the set of Markov strategies is a subset of the set of counter-based strategies that uses space at most $\log T$.

\smallskip\noindent{\em Period of a counter-based strategy.}
We will distinguish between two kinds of memories for a counter-based strategy $\sigma$. One kind is only used once (the initial phase) and the other kind is used arbitrarily many times (the periodic phase). Let $m^0=\emptyset$ and $m^i=\sigma_u(m^{i-1})$. Then if $m^i=m^j$ for some $i<j$, we also have that $m^{i+c}=m^{j+c}$ and $m^{i}=m^{i+c(j-i)}$. Hence if a memory is used twice, it will be reused again. We will let the number of memories that are only used once be $N$ and the number of memories used more than once be $p$, which we will call the period. The number $N$ is mainly important for $\epsilon$-optimal strategies and period is mainly important for optimal strategies.

\smallskip\noindent{\em Probability measure and values.}
A pair of strategies $(\sigma_1,\sigma_2)$, one for each player (in either an SSG or an FSSG), 
defines a probability that the pebble is eventually moved to $\perp$. Let the probability be denoted as 
$P^{\sigma_1,\sigma_2}$.
For all SSGs $G$ (resp. FSSGs $G^T$) it follows from the results of Everett~\cite{Everett} 
that 
\[\sup_{\sigma_1 \in \Pi'_1}\inf _{\sigma_2 \in \Pi_2} P^{\sigma_1,\sigma_2}=\inf _{\sigma_2 \in \Pi'_2}\sup_{\sigma_1 \in \Pi_1} P^{\sigma_1,\sigma_2}.\] 
We will call this common value as the \emph{value} of $G$ (resp. $G^T$) and denote it $\val(G)$ (resp. $\val(G^T)$). 

\smallskip\noindent{\em $\epsilon$-optimal and optimal strategies.}
For all $\epsilon\geq 0$, we will say that a strategy $\sigma_1$ is {\em $\epsilon$-optimal for Player~1} if 
\[\inf _{\sigma_2 \in \Pi_2} P^{\sigma_1,\sigma_2}+\epsilon\geq \sup_{\sigma'_1 \in \Pi'_1} \inf _{\sigma_2 \in \Pi_2} P^{\sigma'_1,\sigma_2}.\] 
Similarly, a strategy $\sigma_2$ is {\em $\epsilon$-optimal for Player~2} if 
\[\sup_{\sigma_1 \in \Pi_1} P^{\sigma_1,\sigma_2}-\epsilon\leq \inf _{\sigma'_2 \in \Pi'_2}\sup_{\sigma_1 \in \Pi_1} P^{\sigma_1,\sigma'_2}.\] 
A strategy $\sigma$ is {\em optimal for Player~$i$} if it is $0$-optimal.  
Condon~\cite{Condon92} showed that {\em there exist optimal memoryless strategies} for any SSG $G$ that are also optimal for $G_s$ for all $s\in S$. 
This also implies that there are optimal Markov strategies for FSSGs that are also optimal for $G_s$ for all $s\in S$. 

\section{Bounds on $\epsilon$-optimal counter-based strategies}\label{sec:counter}

We will first show an upper bound on size of the memory used by a counter-based strategy for playing $\epsilon$-optimal in time limited games. The upper bound on memory size is by application of a result from Ibsen-Jensen and Miltersen~\cite{IJM11}.
The idea of the proof is that if we play an optimal strategy of $G$ in $G^T$ for sufficiently high $T$, then the value we get approaches the value of $G$.

\begin{theorem}\label{thm:upper}{\em (Upper bound)} For all FSSGs $G^T$ with $n$ states and $\epsilon>0$, there is an $\epsilon$-optimal counter-based 
strategy for both players such that memory size is at most $\log \log \epsilon^{-1}+n+1$
\end{theorem}
\begin{proof}
Since there is an optimal Markov strategy, there is a counter-based strategy, which uses memory at most $\log T$.
As shown by Ibsen-Jensen and Miltersen~\cite{IJM11} 
for any game $G^T$, if the horizon is greater than $2 \log \epsilon^{-1}2^n$, the value of $G^T$ approximates the value of $G$
with in $\epsilon$. It is clear that the value of all states are the same in an infinite-horizon game if either player is forced to play an optimal strategy. Hence, if $T\geq 2 \log \epsilon^{-1}2^n$ and  either player plays an optimal strategy of $G$ in $G^T$, then the value of all states are within $\epsilon$ of the value of the game. But there are optimal memoryless strategies in $G$ as shown by Condon~\cite{Condon92}. Therefore we have that in the worst case $T< 2 \log \epsilon^{-1}2^n$. Since $\log T$ is an upper bound, $\log \log \epsilon^{-1}+n+1$ is also an upper bound and hence the result. 
\end{proof}

We will now lower bound the size of the memory needed for a counter-based strategy to be $\epsilon$-optimal. 
Our lower bound will be divided into two parts. The first part will show that $\log \log \epsilon^-1$ is a lower bound on the memory required 
even for some MDPs with constantly many states.
The second part will show that even for fixed $\epsilon$, an $\epsilon$-optimal counter-based strategy will need to use a memory of size $O(n)$. 
Both lower bounds will show explicit MDPs with the required properties. See Figure~\ref{fig:epsilon} and Figure~\ref{fig:h4} respectively.

\begin{figure}
\centering
\begin{tikzpicture}[scale=0.3,->,>=stealth',shorten >=1pt,auto,node distance=4.2cm*0.5,
                    semithick]
      
\tikzstyle{every state}=[fill=white,draw=black,text=black]
\tikzstyle{max}=[state,regular polygon,regular polygon sides=3]
\tikzstyle{min}=[max,regular polygon rotate=180]

    \node[state,accepting] (goal) {$\perp$};
    \node[state] (trap) [below of =goal]{$\top$};
    \node[state] (m1) [right of=goal] {1};
    \node[state] (m2) [right of=m1] {2};
    \node[max] (max) [below right of=m2] {$\maxState$};
        \node[state] (start) [right of=max] {$\start$};
    \node[state] (h) [below of=m1] {$h$};
    
\path 
(m1) edge [bend left] (goal)
(m1) edge [bend right] (goal)
(m2) edge [bend left] (m1)
(m2) edge [bend right] (m1)
(max) edge (m2)
(max) edge (h)
(h) edge (goal)
(h) edge (trap)
(trap) edge [loop above] (trap)
(trap) edge [loop below] (trap)
(start) edge [loop above] (start)
(start) edge  (max)
;

\end{tikzpicture}

\caption{An MDP $G$, such that for all $\epsilon>0$ there is a $T$, such that all $\epsilon$-optimal memory-based strategies for $G^T$ require memory size of at least $\Omega(\log \log \epsilon^{-1})$. Circle vertices are the coin toss states. The triangle vertex is the max state. The vertex $\perp$ is the terminal state. }
\label{fig:epsilon}
\end{figure}
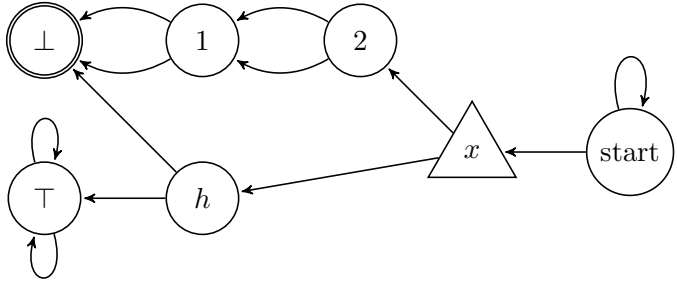

\smallskip\noindent{\em MDP for the lower bound of $\log\log \epsilon^{-1}$.}
Our first lower bound shows that in the MDP $M$ (Figure~\ref{fig:epsilon}) 
all $\epsilon$-optimal memory-based strategies require at least $\log \epsilon^{-1}$ distinct memory states,
i.e., the size of memory is at least $\log\log \epsilon^{-1}$. 
The MDP $M$ is defined as follows. There is one state $\maxState$ in $S_1$, the rest are in $S_R$. 
\begin{itemize}
\item{} The state $\top\in S_R$ has $A_\top=\{(\top,\top),(\top,\top)\}$.
\item{} The state $h\in S_R$ has $A_h=\{(h,\top),(h,\bot)\}$.
\item{} The state $1\in S_R$ has $A_1=\{(1,\perp),(1,\perp)\}$.
\item{} The state $2\in S_R$ has $A_2=\{(2,1),(2,1)\}$.
\item{} The state $\maxState\in S_1$ has $A_m=\{(\maxState,2),(\maxState,h)\}$.
\item{} The state start$\in S_R$ has $A_{\text{start}}=\{(\text{start},\text{start}),(\text{start},\maxState)\}$.
\end{itemize}

\begin{lemma}\label{lem:lowere}
All $\epsilon$-optimal memory-based strategies in $M^{T}$, for $T=\log \epsilon^{-1} - 1$, require at least $\log \epsilon^{-1}-2$ 
distinct states of memory, i.e., the size of memory is at least $\log\log \epsilon^{-1}$. 
\end{lemma}

\begin{proof}
We will first show the proof for counter-based strategies. At the end we will then extend it to memory-based strategies.

It is clear that $\val(M^2_{\maxState})=\frac{1}{2}$ and  for all $T>2$ we have $\val(M^T_{\maxState})=1$. If Player~$1$ chooses $(\maxState,h)$ in $M^2_{\maxState}$, then he gains $\frac{1}{2}$, otherwise, if he chooses $(\maxState,2)$, then he gains $0$. Also for all $T>2$, if Player~1 chooses $(\maxState,2)$ in $M^T_{\maxState}$, then he gains $1$, otherwise, if he chooses $(\maxState,h)$, then he gains $\frac{1}{2}$.

In $M_{\start}$ we end up at $\maxState$ after precisely $k\geq 2$ moves of the pebble with probability $2^{-k+1}$. Therefore, by the preceding any optimal memory-based strategy $\sigma$ must be able to find out if $T$ minus the length of the history is greater than $2$ from the memory. 

Let $\epsilon>0$ be given. For simplicity we will assume that $\epsilon=2^{-k}$ for some $k>0$. Let $c=\log \epsilon^{-1}$. Assume now that there is a  counter-based strategy $\sigma=(\sigma_u,\sigma_a)$ that uses  $c-3$ states of memory in $M^{c-1}_{\start}$. The pebble ends up at $m$ after $c-3$ moves with probability $2^{-(c-3)+1}=4\epsilon$. Let the sequences of memories until then be $m^0,m^1,\dots,m^{c-3}$. Since $\sigma$ was $\epsilon$-optimal we must have that $\sigma(m^{c-3},\maxState)=(\maxState,h)$. On the other hand for all $i<c-3$  we must also have that $ \sigma(m^i,\maxState)=(\maxState,2)$. Therefore $m^{c-3}$ differs from $m^i$ for $i< c-3$. Now assume that $m^i=m^j$ for $i<j$ and $i,j< c-3$. But then $\sigma_u(m^i)=\sigma_u(m^j)$ and hence $m^{i+1}=m^{j+1}$ and then by repeating this argument we have that $m^{k}=m^{c-3}$ for $k<c-3$. Therefore $m^i$ differs from $m^j$ for $i\neq j$ and $i,j\leq c-3$ and hence we need at least $c-2$ different memory states.

For general memory-based strategies the proof remains the same. This is because we can note that if the pebble ends up at $\maxState$ after $c-3$ moves, we have that $m^0=\emptyset$ and $m^i=\sigma_u(m^{i-1},\start)$ for $1\leq i\leq c-3$ and hence they must all differ by the same argument as before.
\end{proof}

For our second lower bound we will use an infinite family of MDPs \[H=\{H(1),H(2),\dots,H(i),\dots\},\] such that $H(i)$ contains $2i+4$ states, one of which is a max state, and all $\epsilon$-optimal counter-based strategies require space at least $i-4$, for some fixed $\epsilon$.

\smallskip\noindent{\em Family of MDPs for the lower bound of $n$.}
The MDP $H(i)$ is defined as follows. There is one state $\maxState$ in $S_1$, the rest are in $S_R$. \begin{itemize}
\item{} The state $\top\in S_R$ has $A_\top=\{(\top,\top),(\top,\top)\}$.

\item{} The state $h\in S_R$ has $A_h=\{(h,\top),(h,\bot)\}$.
\item{} The state $1\in S_R$ has $A_1=\{(1,\perp),(1,i)\}$.
\item{} For $j\in \{2,\dots,i\}$, the state $j\in S_R$ has $A_j=\{(j,i),(j,j-1)\}$.
\item{} The state $\maxState\in S_1$ has $A_m=\{(\maxState,i),(\maxState,h)\}$.
\item{} The state $1^*\in S_R$ has $A_{1^*}=\{(1^*,i^*),(1^*,\maxState)\}$.
\item{} For $j\in \{2,\dots,i\}$, the state $j^*\in S_R$ has $A_{j^*}=\{(j^*,i^*),(j^*,(j-1)^*)\}$.
\end{itemize}
There is a illustration of $H(4)$ in Figure \ref{fig:h4}.

Let $i$ be some number. It is clear that $\val(H(i)^2_{\maxState})=\frac{1}{2}$. It is also easy to see that $\val(H(i)_i)=1$, but that the time to reach $\perp$ from $i$ is quite long. Hence, one can deduce that there must be a $k$ ($k$ depends on $i$) such that for all $k'\geq k$ it is an optimal strategy in $H(i)_{\maxState}^{k'}$ to choose $(\maxState,i)$ and for all $2\leq k''<k$ it is an optimal strategy in $H(i)_{\maxState}^{k''}$ to choose $(\maxState,h)$. In case there are multiple such numbers, let $k$ be the smallest. The number $k-1$ is then the smallest number of moves of the pebble to reach $\perp$ from $i$, such that that occurs with probability  $\geq \frac{1}{2}$ (to simplify the proofs we will assume equality). 

Let $p^t$ be the probability for the pebble to reach $\maxState$ from $i^*$ in $t$ or less moves (note that this is also the probability to reach $\perp$ in $t$ moves or less from $i$). It is clear that $p^t$ is equal to the
probability that a sequence of $t$ fair coin tosses contains $i$ consecutive tails. This is known to be exactly
$1-F_{t+2}^{(i)}/2^t$, where $F_{t+2}^{(i)}$ is the $(t+2)$'nd Fibonacci
$i$-step number, i.e. the number given by the linear homogeneous recurrence
$F_c^{(i)} = \sum_{j=1}^{i} F_{c-j}^{(i)}$ and the
boundary conditions $F_{c}^{(i)} = 0$, for $c \leq 0$,
$F_{1}^{(i)} = F_{2}^{(i)} = 1$ (this fact is also mentioned in Ibsen-Jensen and Miltersen~\cite{IJM11}).

The next lemmas will prove various properties of $p^t$, $F^{(i)}_a$ and $k$. We will first show two technical lemmas that will be used in many of the remaining lemmas. Next, we will show that $k$ is exponential in $i$ and show various bounds on $p^t$. We will use all that to show that the number of states in the game is a lower bound on the memory requirement for $\epsilon$-optimal counter-based strategies.

\begin{lemma}\label{lem:Fai}
Let $i$ and $a\geq i+3$ be given. Then \[F_a^{(i)}\leq (2-2^{-i-1})F_{a-1}^{(i)}\]
Let $b\geq 3$ be given. Then \[F_b^{(i)}\leq 2F_{b-1}^{(i)}\]
\end{lemma}
\begin{proof} We can see that \[F_b^{(i)}=\sum_{j=1}^{i} F_{b-j}^{(i)}=2F_{b-1}^{(i)}-F_{b-i}^{(i)}\] for $b\geq 3$. Hence we have that $F_b^{(i)}\leq 2F_{b-1}^{(i)}$. 

We therefore have that $F_{a-i}^{(i)}\geq 2^{-i-1}F_{a-1}^{(i)}$ and we can deduce that \[F_a^{(i)}\leq 2F_{a-1}^{(i)}-2^{-i-1}F_{a-1}^{(i)}.\]
The desired result follows.
\end{proof}

Now for the proof that $k$ is exponential in $i$.

\begin{lemma}\label{lem:kexp}
For all $i$, we have that $k\geq 2^{i-2}+i$. 
\end{lemma}
\begin{proof}
We will first show that $p^a\leq p^{a-1}+2^{-i}$. We can divide the event that there are $i$ consecutive tails into two possibilities out of $t$ fair coin tosses. Either the first $i$ coin tosses were tails or there are $i$ consecutive tails in the last $t-1$ coin tosses (or both). The first case happens with probability $2^{-i}$ and the last with probability $p^{a-1}$. We can then apply union bounds and get that $p^a\leq p^{a-1}+2^{-i}$. 
Clearly we have that $p^{i-1}=0$ and that $p^a$ is increasing in $a$. But we also have that \[
\begin{split}
p^k& \leq 2^{i-2} 2^{-i} + p^{k-2^{i-2}}\Rightarrow\\
\frac{1}{2} & \leq \frac{1}{4} + p^{k-2^{i-2}}\Rightarrow\\
\frac{1}{4} & \leq p^{k-2^{i-2}},
\end{split}
\]
which means that $k> 2^{i-2}+i-1$.
\end{proof}

\begin{lemma}\label{lem:k}
Let $i$ be given. The number $k$ is such that 
\[e^{\frac{-1}{8}} \geq (1-2^{-i-2})^{k}\geq \frac{1}{4}\]
and such that
\[e^{\frac{-1}{8}} \geq (1-2^{-i-2})^{k-i}\geq \frac{1}{2}\]
\end{lemma}

\begin{proof}
We have that $1-F_{k+2}^{(i)}/2^k=\frac{1}{2}$, which we can then use to show that \[
\begin{split}
1-F_{k+2}^{(i)}/2^k&=\frac{1}{2}\Rightarrow\\
F_{k+2}^{(i)}/2^k&=\frac{1}{2}\Rightarrow\\
\frac{(2-2^{-i-1})^{k-i}2^i F_{2}^{(i)}}{2^k}&\geq \frac{1}{2}\Rightarrow\\
(1-2^{-i-2})^{k-i}&\geq \frac{1}{2}
\end{split}
\]
where we used Lemma \ref{lem:Fai} for the second implication. We used that $F_{2}^{(i)}=1$ for the third implication.
Since $k\geq 2^{i-2}+i> 2i$ by Lemma \ref{lem:kexp}, we also have that $(1-2^{-i-2})^{k}\geq \frac{1}{4}$.

But we can also use Lemma \ref{lem:kexp} more directly. Notice that since $i\geq 12$ we have that $2^{i+2}\geq 72$. We have that, \[
(1-2^{-i-2})^{k-i}\leq (1-2^{-i-2})^{2^{i-2}}=((1-2^{-i-2})^{2^{i+2}})^{\frac{1}{8}}\leq e^{\frac{-1}{8}},
\]
where we used that $\lim_{x\rightarrow \infty} (1-x^{-1})^x=e^{-1}$ and that $(1-x^{-1})^x$ is increasing in $x$ for $x\geq 1$. We also have that $e^{\frac{-1}{8}} \geq (1-2^{-i-2})^{k}$, by the same argument.
\end{proof}

\begin{lemma}\label{lem:half t}
For all $i$ and $t$, we have \[
p^{2t-2i}\leq 2p^t
\]
\end{lemma}
\begin{proof}
Let $t'=t-i$. Hence, we need to show that $p^{2t'}\leq 2p^{t'+i}$.
The proof comes from the fact that to have $i$ consecutive tails out of $2t'$ fair coin tosses, the $i$ consecutive tails must either start in the first half or end in the second half (or both). But to start in the first half means that it must end in the first $t'+i$ elements. Therefore we can overestimate that probability with $p^{t'+i}$. Similar with the second half. We can then add them together by union bound and the result follows.
\end{proof}

\begin{lemma}\label{lem:upper bound on pdk}
Let $i\geq 12$ and $\frac{1}{10}<d<1$ be given. Then $p^{dk}\leq 1-\frac{e^{\frac{1-d}{8}}}{2}<\frac{1}{2}$.
\end{lemma}
\begin{proof}
Since $d>\frac{1}{10}$, we have that $dk>i$, by Lemma \ref{lem:kexp} and because $i\geq 12$.
We will show that $F_{dk+2}^{(i)}/2^{dk}\geq \frac{e^{\frac{1-d}{8}}}{2}$.
We have that \[
\begin{split}
F_{dk+2}^{(i)}/2^{dk}& \geq \frac{F_{k+2}^{(i)}}{(2-2^{-i-1})^{(1-d)k} 2^{dk}}\\
& =\frac{F_{k+2}^{(i)}}{(1-2^{-i-2})^{(1-d)k} 2^{k}}\\
& =\frac{1}{2\cdot (1-2^{-i-2})^{(1-d)k}}\\
& =\frac{1}{2\cdot ((1-2^{-i-2})^{k})^{1-d}}\\
& \geq \frac{1}{2\cdot (e^{-\frac{1}{8}})^{1-d}}\\
& = \frac{e^{\frac{1}{8}(1-d)}}{2}
\end{split}
\]
where we used Lemma \ref{lem:Fai} for the first inequality, Lemma \ref{lem:k} for the second and that $\lim_{x\rightarrow \infty} (1-x^{-1})^x=e^{-1}$ and that $(1-x^{-1})^x$ is increasing in $x$ for $x\geq 1$ for the third.
\end{proof}


\begin{lemma}\label{lem:greater than k}
Let $i\geq 12$ and $0<d$ be given. Then $p^{(1+d)k}\geq 1-(e^{\frac{-d}{8}}) \frac{1}{2}>\frac{1}{2}$.
\end{lemma}

\begin{proof}
We will show that $F_{(1+d)k+2}^{(i)}/2^{(1+d)k}\leq (e^{\frac{-d}{8}}) \frac{1}{2}$.
We have that \[
\begin{split}
F_{(1+d)k+2}^{(i)}/2^{dk}& \leq \frac{(2-2^{-i-1})^{dk}F_{k+2}^{(i)}}{2^{(1+d)k}} \\
& = (1-2^{-i-2})^{dk}\frac{1}{2}\\
& = ((1-2^{-i-2})^{k})^d \frac{1}{2}\\
& \leq (e^{\frac{-1}{8}})^d \frac{1}{2}\\
& = (e^{\frac{-d}{8}}) \frac{1}{2}
\end{split}
\]
where we used Lemma \ref{lem:Fai} for the first inequality and Lemma \ref{lem:k} for the second.
\end{proof}

\begin{figure}
\centering
\begin{tikzpicture}[scale=0.3,->,>=stealth',shorten >=1pt,auto,node distance=4.2cm*0.5,
                    semithick]
      
\tikzstyle{every state}=[fill=white,draw=black,text=black]
\tikzstyle{max}=[state,regular polygon,regular polygon sides=3]
\tikzstyle{min}=[max,regular polygon rotate=180]

    \node[state,accepting] (goal) {$\perp$};
    \node[state] (trap) [left of =goal]{$\top$};
    
    \node[state] (m1) [below of=goal] {1};
    \node[state] (m2) [below of=m1] {2};
    \node[state] (m3) [below of=m2] {3};
    \node[state] (m4) [below of=m3] {4};
    
    \node[max] (max) [below left of=m4] {$\maxState$};
    \node[state] (h) [left of=m1] {$h$};
    
        \node[state] (m1s) [left of=max] {$1^*$};
    \node[state] (m2s) [above of=m1s] {$2^*$};
    \node[state] (m3s) [above of=m2s] {$3^*$};
    \node[state] (m4s) [above of=m3s] {$4^*$};
    
\path 
(m1) edge (goal)
(m1) edge [bend left,in=120] (m4)
(m2) edge (m1)
(m2) edge [bend right] (m4)
(m3) edge (m2)
(m3) edge [bend left] (m4)
(m4) edge (m3)
(m4) edge [loop left] (m4)
(max) edge (m4)
(max) edge (h)
(h) edge (goal)
(h) edge (trap)
(trap) edge [loop left] (trap)
(trap) edge [loop right] (trap)

(m1s) edge (max)
(m1s) edge [bend left,in=120] (m4s)
(m2s) edge (m1s)
(m2s) edge [bend right] (m4s)
(m3s) edge (m2s)
(m3s) edge [bend left] (m4s)
(m4s) edge (m3s)
(m4s) edge [loop right] (m4)

;

\end{tikzpicture}
\caption{The MDP $H(4)$. It is the fourth member of a family that will show that there exist FSSGs where, for a fixed $\epsilon$, all $\epsilon$-optimal counter-based strategies require memory size to be at least $\Omega(i)$. Circle vertices are the coin toss states. The triangle vertex is the max state. The vertex $\perp$ is the terminal state. }
\label{fig:h4}
\end{figure}
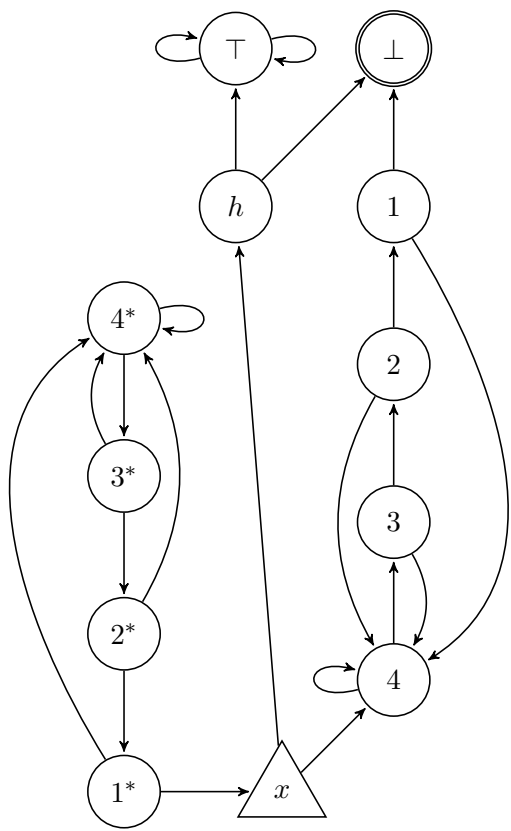

\begin{lemma}
There is an $\epsilon$ such that for all $i\geq 12$, there is a time-bound $T$ such that all $\epsilon$-optimal counter-based strategies for $H(i)^T$ require memory size at least $i-5$.\label{lem:lowern}
\end{lemma}

The proof basically goes as follows: The pebble starts at $i^{*}$ with $2k+1$ moves remaining. First we show that there is a super-constant probability for the pebble to reach $\maxState$ using somewhere between $\frac{k}{5}$ and $\frac{4k}{5}$ moves. In that case there is at least $\frac{6k}{5}+1$ moves left. We then show that there is some number $p>\frac{1}{2}$ independent of $i$ such that the probability to reach $\perp$ from $i$ in $\frac{6k}{5}$ is more than $p$. Secondly we show that there is a super-constant probability for the pebble to reach $\maxState$ using somewhere between $\frac{6k}{5}$ and $\frac{9k}{5}$ moves.  In that case there is at most $\frac{4k}{5}+1$ moves left. We then show that there is some number $q<\frac{1}{2}$ independent of $i$ such that the probability to reach $\perp$ from $i$ in $\frac{4k}{5}$ is less than $q$. We can then pick $\epsilon$ such that any $\epsilon$-optimal strategy must distinguish between plays that used between $\frac{k}{5}$ and $\frac{4k}{5}$ moves to reach $\maxState$ from $i^{*}$ and plays that used between $\frac{6k}{5}$ and $\frac{9k}{5}$ moves to reach $\maxState$ from $i^{*}$. We then show that that requires at least $O(k)$ distinct states of memory, and the result then follows from $k$ being exponential in $i$, by Lemma \ref{lem:kexp}.

\begin{proof}The probability  for the pebble to reach $\maxState$ using somewhere between $\frac{k}{5}$ and $\frac{4k}{5}$ moves is 
\[
\begin{split}
p^{\frac{4k}{5}}-p^{\frac{k}{5}}&=1-F_{\frac{4k}{5}+2}^{(i)}/2^{\frac{4k}{5}}-(1-F_{\frac{k}{5}+2}^{(i)}/2^{\frac{k}{5}})\\
&=\frac{2^{\frac{3k}{5}}F_{\frac{k}{5}+2}^{(i)}-F_{\frac{4k}{5}+2}^{(i)}}{2^{\frac{4k}{5}}}\\
&\geq \frac{2^{\frac{3k}{5}}F_{\frac{k}{5}+2}^{(i)}-(2-2^{-i-1})^{\frac{3k}{5}}F_{\frac{k}{5}+2}^{(i)}}{2^{\frac{4k}{5}}}\\
&= \frac{(2^{\frac{3k}{5}}-(2-2^{-i-1})^{\frac{3k}{5}})F_{\frac{k}{5}+2}^{(i)}}{2^{\frac{4k}{5}}}\\
&= \frac{(1-(1-2^{-i-2})^{\frac{3k}{5}})F_{\frac{k}{5}+2}^{(i)}}{2^{\frac{k}{5}}}\\
&= (1-(1-2^{-i-2})^{k})(1-p^{\frac{k}{5}})\\
&\geq (1-e^{\frac{-1}{8}} )\frac{e^{\frac{1}{10}}}{2}
\end{split}
\]
where we used Lemma \ref{lem:Fai} for the first inequality and Lemma \ref{lem:k} and Lemma \ref{lem:upper bound on pdk} for the second.

In this case we have at least $\frac{6k}{5}+1$ moves left. Therefore if the player chooses to move to $i$, there are at least $\frac{6k}{5}$ moves left. In that case, by Lemma \ref{lem:greater than k}, the pebble will reach $\perp$ with probability at least $1-(e^{\frac{-3}{40}})\frac{1}{2}>\frac{1}{2}$. In both cases we see that the probability is strictly separated from $\frac{1}{2}$.

The probability  for the pebble to reach $\maxState$ using somewhere between $\frac{6k}{5}$ and $\frac{9k}{5}$ moves can be calculated similar to between  $\frac{k}{5}$ and $\frac{4k}{5}$ moves. We end up with
\[
p^{\frac{9k}{5}}-p^{\frac{6k}{5}}\geq (1-e^{\frac{-1}{8}} )(1-p^{\frac{6k}{5}}).
\]

Hence, we need a upper bound on $p^{\frac{6k}{5}}$, which is smaller than 1 and does not depend on $k$ or $i$. We can get that by noting that $\frac{6k}{5}\leq \frac{8k}{5}-2i$, because of Lemma \ref{lem:kexp} and that $i\geq 12$. Hence we can apply Lemma \ref{lem:half t} followed by Lemma \ref{lem:upper bound on pdk} and get that $p^{\frac{6k}{5}}\leq 2p^{\frac{4k}{5}}\leq 2(1-\frac{e^{\frac{1}{40}}}{2})<1$.

In this case we have at most $\frac{4k}{5}+1$ moves left. Therefore if the player chooses to move to $i$, there are at most $\frac{4k}{5}$ moves left. In that case, by Lemma \ref{lem:upper bound on pdk}, the pebble will reach $\perp$ with probability at most $1-\frac{e^{\frac{1}{40}}}{2}<\frac{1}{2}$.

Let $\sigma$ be some $\epsilon$-optimal counter-based strategy and assume that $\sigma$ uses less than $\frac{k}{5}-1$ states. We will show that if $\epsilon$ is some sufficiently low constant, we get a contradiction and hence all $\epsilon$-optimal counter-based strategies uses at least $\frac{k}{5}$ states. Our result than follows from Lemma \ref{lem:kexp}.

Let $m^0=\emptyset$ and $m^i=\sigma_u(m^{i-1})$. Since $\sigma$ uses less than $\frac{k}{5}$ states, then $m^a=m^b$ for some $a<b<\frac{k}{5}$. Hence also $m^{a+c}=m^{b+c}$ for all $c\geq 0$, by definition. But then $m^{a+c}=m^{a+c+(b-a)d}$ for all $c$ and $d$ greater than 0. Hence, we can make a one to one map between memory  $m^a$ for $a\in A=\{\frac{k}{5},\dots,\frac{4k}{5}\}$ and some memory  $m^b$ for $b\in B=\{\frac{6k}{5},\dots,\frac{9k}{5}\}$, such that $m^a=m^b$, except for up to $\frac{k}{5}$ of them, which is smaller than a third of the size of both $A$ and $B$. 

Let $q^t$ be the the probability to reach $\maxState$ from $i^{*}$ using exactly $t$ moves of the pebble. For $t\geq i+1$ we have that \[q^t=p^t-p^{t-1}=\frac{2F_{t+1}^{(i)}-F_{t+2}^{(i)}}{2^t}=\frac{F_{t+1-i}^{(i)}}{2^t}=2^{-i-1}(1-p^{t-1-i}).\] 
(To have a sequence of $i$ tails after precisely $t$ coin flips for $t>i$, we need to have failed to get that many tails in a row for the first $t-1-i$ coin flips and then gotten a head followed by $i$ tails, which is also what our expression tells us.)

We see that $q^t$ is decreasing for $t\geq i+1$, because $p^t$ is increasing. We can therefore calculate the probability to end up at $\maxState$ using a specific amount of time compared to all other times in $A$ as 
\[
\begin{split}
\frac{q^{\frac{k}{5}}}{q^{\frac{4k}{5}}}&=\frac{2^{-i-1}(1-p^{\frac{k}{5}-1-i})}{2^{-i-1}(1-p^{\frac{4k}{5}-1-i})}\\
&=\frac{\frac{F_{\frac{k}{5}+1-i}^{(i)}}{2^{\frac{k}{5}-1-i}}}{\frac{F_{\frac{4k}{5}+1-i}^{(i)}}{2^{\frac{4k}{5}-1-i}}}\\
&\geq \frac{ F_{\frac{k}{5}+1-i}^{(i)}2^{\frac{3k}{5}}}{(2-2^{-i-1})^{\frac{3k}{5}}F_{\frac{k}{5}+1-i}^{(i)}}\\
&= (1-2^{-i-2})^{-\frac{3k}{5}}\\
&= ((1-2^{-i-2})^{k})^{-\frac{3}{5}}\\
&\geq e^{\frac{3}{40}},
\end{split}
\]
where we used Lemma \ref{lem:Fai} for the first inequality and Lemma \ref{lem:k} for the second.

We can show similarly that all $q^t$ for $t$ being in $B$ are also equal up to a factor of $e^{\frac{3}{40}}$. 
Hence, the probability to reach $\maxState$ from $i^{*}$ with $t$ time remaining for $t-1\in A$ is nearly uniformly distributed over $A$ (up to a factor of $e^{\frac{3}{40}}$). Similar with $t-1$ in $B$.
Therefore we can pick an $\epsilon_1$ (independent of $i$) such that $\sigma_a(m^t,\maxState)=(\maxState,h)$ for all but $\frac{1}{10}$ of the $t$'s in $A$. Similar, we can pick an $\epsilon_2$ (independent of $i$) such that $\sigma_a(m^t,\maxState)=(\maxState,i)$ for all but $\frac{1}{10}$ of the $t$'s in $B$.

By using $\epsilon=\min(\epsilon_1,\epsilon_2)$ both $\frac{9}{10}$ of all $t$ in  $A$ have that $\sigma_a(m^t,\maxState)=(\maxState,h)$ and $\frac{9}{10}$ of all $t$ in   $B$ have that $\sigma_a(m^t,\maxState)=(\maxState,i)$. But this contradicts that we had a one to one map that mapped at least two thirds of all $m^a$ for $a$ in $A$ to some  memory $m^{b}$ for $b$ in $B$ such that $m^a=m^b$ (and at least two thirds of the $b'$s got mapped to).

Hence all $\epsilon$-optimal counter-based strategies uses memory at least $\frac{k}{5}$. The result then follows from $k\geq 2^{i-2}+i$ from Lemma \ref{lem:kexp}.
\end{proof}


\begin{theorem}\label{thm:lower}{\em (Lower bound)} For all sufficiently small $\epsilon>0$ and all $n\geq 5$, there is a FMDP with $n$ states, where all $\epsilon$-optimal counter-based strategies require memory size at least $\Omega(\log \log \epsilon^{-1}+n)$.
\end{theorem}

\begin{proof}
The proof is a simple combination of the two lower bounds in Lemma \ref{lem:lowere} and Lemma \ref{lem:lowern}. 
\end{proof}




\section{A lower bound on the period of optimal strategies in MDPs}\label{sec:period}

\begin{figure}
\centering
\begin{tikzpicture}[scale=0.3,->,>=stealth',shorten >=1pt,auto,node distance=4.2cm*0.5,
                    semithick]
      
\tikzstyle{every state}=[fill=white,draw=black,text=black]
\tikzstyle{max}=[state,regular polygon,regular polygon sides=3]
\tikzstyle{min}=[max,regular polygon rotate=180]

    \node[state,accepting] (goal) {$\perp$};
    \node[state] (m1) [below of=goal] {1*};
    \node[state] (m2) [below of=m1] {2*};
    \node[state] (m3) [below of=m2] {3*};
    \node[state] (m4) [below of=m3] {4*};
    
    \node[state] (q5) [below of=m4] {5};
    \node[state] (q4) [right of=q5] {4};
    \node[state] (q3) [right of=q4] {3};
    \node[state] (q2) [right of=q3] {2};
    \node[state] (q1) [right of=q2] {1};
    \node[max](max)[below of=q2]{};

\path 
(m1) edge [bend left] (goal)
(m1) edge [bend right] (goal)
(m2) edge [bend left] (m1)
(m2) edge [bend right] (m1)
(m3) edge [bend left] (m2)
(m3) edge [bend right] (m2)
(m4) edge [bend left] (m3)
(m4) edge [bend right] (m3)
(q5) edge (m4)
(q5) edge (q4)
(q4) edge [in=-25,out=90](m3)
(q4) edge (q3)
(q3) edge [in=-25,out=90] (m2)
(q3) edge (q2)
(q2) edge [in=-25,out=90] (m1)
(q2) edge (q1)
(q1) edge [in=-25,out=90] (goal)
(q1) edge [bend left] (q5)
(max) edge (q1)
(max) edge (q2)
;

\end{tikzpicture}
\caption{The MDP $G_5$. Circle vertices are the coin toss states. The triangle vertex is the Max state. The vertex $\perp$ is the terminal state. }
\label{fig:g5}
\end{figure}
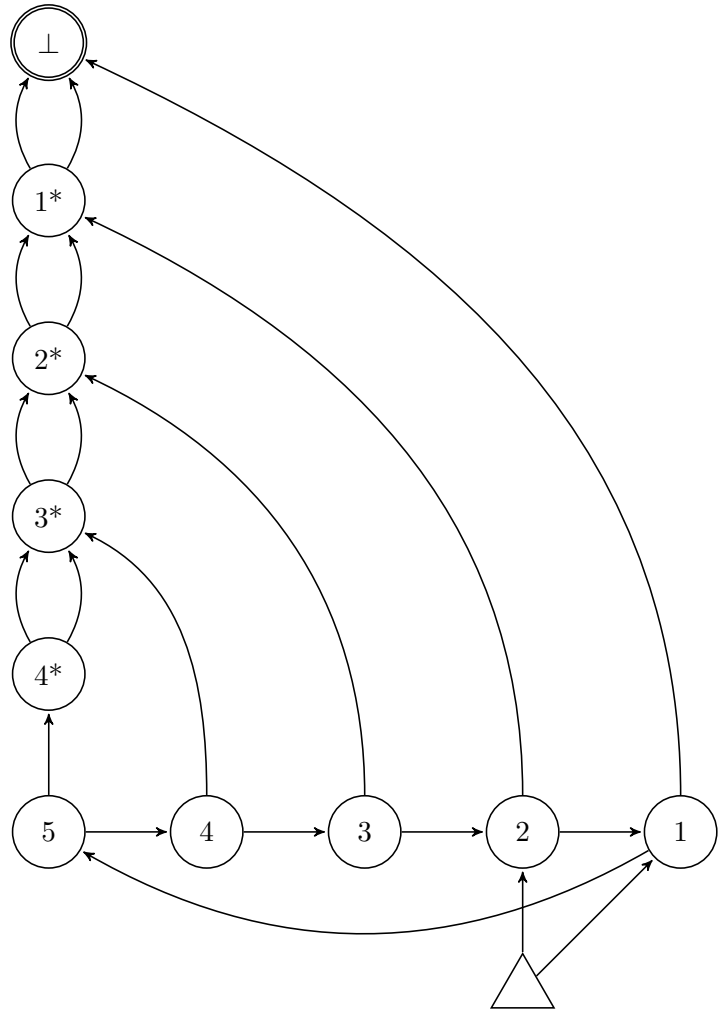

We will in this section show that there exist FMDPs $G$, with $n$ states, such that all optimal strategies can be implemented using a counter-based strategy, and the period is greater than $2^{\Omega(\sqrt{n\log n})}$. We will create such FMDPs in two steps. First we will construct a family, such that the $i'$th member requires that one state uses one action every $\Theta(i)$ steps and in all other steps uses the other action. There is an illustration of a member of that family in Figure \ref{fig:g5}. Afterwards we will play many such games in parallel, which will ensure that a large period is needed for all optimal strategies. There is an illustration of such a game in Figure \ref{fig:f2}.

Let $G_p$, $p\in \{2,3,\dots\}$ be the following FMDP, with $2p-1$ coin toss states and one max state. The coin toss states are divided into the sets $\{1^*,2^*,\dots,(p-1)^*\}$ and $\{1,2,\dots,p\}$. To simplify the following description let state $0^*$ denote the $\perp$ terminal state. A description of $G$ is then \begin{itemize}
\item{}State $i^*$ has state $(i-1)^*$ as both its successors. 
\item{} State $i$ has state $(i-1)^*$ and $(i-1)$ as successors, except state $1$ which has $\perp$ and state $p$ as successors. 
\item{} The max state has $1$ and $2$ as successors.
\end{itemize}
There is an illustration of $G_5$ in Figure \ref{fig:g5}.

\begin{lemma}\label{value of gp}
Let $p\geq 2$ be given. State $i$ has value $1-2^{-f_i(k)}$ in $G_p^{k}$ for $k>0$, where $f_i(k)$ is the function $f_i(k)=\max_{k'\leq k\wedge k'\modS p=i}(k',0)$. 
\end{lemma}
\begin{proof}
It is easily seen by induction that $i^*$ has value 1 in $G_p^i$.
Note that $f_i(k)=i$ for $k\modS p=i$.
The proof will be by induction in $k$. There will be one base case and two induction cases, one for $1< k\leq p$ and one for $k>p$. 
It is easy to see that state $1$ has value $\frac{1}{2}=1-\frac{1}{2}=1-2^{-f_1(1)}$ in $G_p^1$ and state $j$ for $j\neq 1$ has value 0. That settles the base case.

For $1< k\leq p$. Neither of the successors of state $j$, for $j\neq k$, has changed values from $G_p^{k-2}$ to $G_p^{k-1}$. For state $k$, both its successors has changed value. The value of state $k-1^*$ has become $\val(G_p^{k-1})_{k-1^*}=1$ and the value of state $k-1$ has become $\val(G_p^{k-1})_{k-1}=1-2^{-f_{k-1}(k-1)}$. The value of state $k$ is then \[\val(G_p^{k})_{k}=\frac{1+1-2^{-f_{k-1}(k-1)}}{2}=\frac{1+1-2^{-(k-1)}}{2}=1-2^{-(k-1)-1} =1-2^{-f_k(k)}.\]

For $p<k$. Let $i$ be $k\modS_p$. Neither of the successors of state $j$, for $j\neq i$, has changed values from $G_p^{k-2}$ to $G_p^{k-1}$. The value of state $i'=i-1\modS_p$, in iteration $k-1$ is $\val(G_p^{k-1})_{i'}=1-2^{-f_{i'}(k-1)}$. The value of state $i$ is then \[\val(G_p^{k})_{i}=\frac{1+1-2^{-f_{i'}(k-1)}}{2}=\frac{1+1-2^{-(k-1)}}{2}=1-2^{-(k-1)-1} =1-2^{-f_i(k)}.\]
The desired result follows.
\end{proof}

The idea behind the construction of $F_k$ is that to find the state of the largest value among 1 and 2, in $G^T_p$, for $p\geq 2$ and $T\geq 1$, we need to know if $T\modS p=1$ or not.
Let $p_i$ be the $i$'th smallest prime number. The FMDP $F_k$ is as follows:
$F_k$ consists of a copy of $G_{p_i}$ for $i\in \{1,\dots,k\}$. Let the max state in that copy of $G_{p_i}$ be $m_i$. 
There is a illustration of $F_2$ in Figure \ref{fig:f2}.

\begin{figure}
\centering
\begin{tikzpicture}[scale=0.3,->,>=stealth',shorten >=1pt,auto,node distance=4.2cm*0.5,
                    semithick]
      
\tikzstyle{every state}=[fill=white,draw=black,text=black]
\tikzstyle{max}=[state,regular polygon,regular polygon sides=3, inner sep=-0.05cm]
\tikzstyle{min}=[max,regular polygon rotate=180]

    \node[state,accepting] (goal) {$\perp$};
    \node[state] (m1n2) [below right of=goal] {};

    \node[state] (q2n2) [below of=m1n2] {};
    \node[state] (q1n2) [right of=q2n2] {};
    \node[max](maxn2)[below of=q2n2]{$m_1$};

    \node[state] (m1n3) [below left of=goal] {};
    \node[state] (m2n3) [below of=m1n3] {};

    \node[state] (q3n3) [below of=m2n3] {};
    \node[state] (q2n3) [left of=q3n3] {};
    \node[state] (q1n3) [left of=q2n3] {};
    \node[max](maxn3)[below of=q2n3]{$m_2$};

\path 
(m1n2) edge [bend left] (goal)
(m1n2) edge [bend right] (goal)
(q2n2) edge (m1n2)
(q2n2) edge (q1n2)
(q1n2) edge [in=0,out=90]  (goal)
(q1n2) edge [bend left] (q2n2)
(maxn2) edge (q1n2)
(maxn2) edge (q2n2)

(m1n3) edge [bend left] (goal)
(m1n3) edge [bend right] (goal)
(m2n3) edge [bend left] (m1n3)
(m2n3) edge [bend right] (m1n3)
(q3n3) edge (m2n3)
(q3n3) edge (q2n3)
(q2n3) edge [in=180,out=90]  (m1n3)
(q2n3) edge (q1n3)
(q1n3) edge [in=180,out=90]  (goal)
(q1n3) edge [bend right] (q3n3)
(maxn3) edge (q1n3)
(maxn3) edge (q2n3)
;

\end{tikzpicture}
\caption{The FMDP $F_2$. Circle vertices are the coin toss states. Triangle vertices are the max states. The vertex $\perp$ is the terminal state. }
\label{fig:f2}
\end{figure}
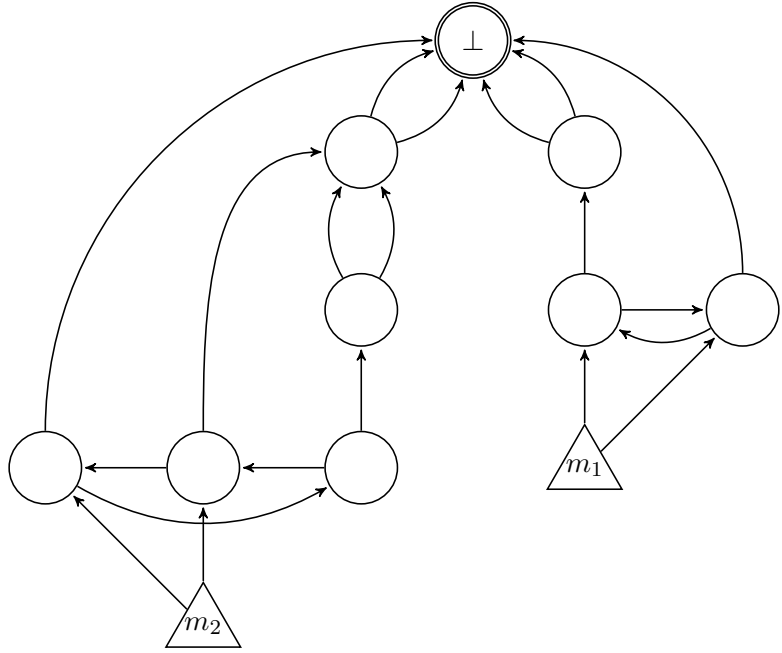

We will now show that all optimal strategies  for $F_k$ are subsets of counter-based strategies with a period defined by $k$. Afterwards we will show that the number of states in $F_k$ can also be expressed in terms of $k$. At the end we will use those two lemmas to get to our result.

\begin{lemma}\label{lem:bits needed for Fk}
Any optimal strategy $\sigma(k,T')$ in $F_k$ is an finite memory counter-based strategies with period $P=\prod_{i\in \{1,\dots,k\}} p_i$, where $p_i$ is the $i$'th smallest prime number.
\end{lemma}
\begin{proof}
Let $i$ be some number in $\{1,\dots,k\}$. The lone optimal choice for $m_i$ and $T'>0$ is to use the action that goes to state $1$ in $G_{p_i}$ if $T\modS p_i=1$ and otherwise to use the action that goes to state $2$ in $G_{p_i}$ by Lemma \ref{value of gp}. Hence, by the Chinese remainder theorem there are precisely $P$ steps between each time any optimal strategy uses the action that goes to $1$ in all $m_i$'s. That is, any optimal strategy must do the same action at least every $P$ steps. Furthermore it is also easy to see that any optimal strategy must do the same at most every $P$ steps, by noting that $T+P\modS p_i$ is 1 if and only if $T\modS p_i$ is 1 and again applying Lemma \ref{value of gp}. A strategy that does the same every $P$ steps can be expressed by a counter-based strategy with period $P$, which also uses memory at most $P$.
\end{proof}

\begin{lemma}\label{lem:states in Fk}
The number of states in $F_k$ is $2 \sum_{i\in \{1,\dots,k\}} p_i$.
\end{lemma}
\begin{proof}
For any $i$, $G_{p_i}$ consists of $2p_i$ states. $F_k$ therefore consists of $2 \sum_{i\in \{1,\dots,k\}} p_i$ states.
\end{proof}

\begin{theorem}\label{thm:period}
There are FMDPs $G$, with $n$ states, where all optimal strategies are finite memory counter-based strategies with period  $2^{\Omega(\sqrt{n\log n})}$.
\end{theorem}
\begin{proof}
Let $n$ be such that there exists a game $F_k$ with $n$ states. Note that for any number there is always a larger number, $a$, such that $F_b$ has $a$ states for some $b$.

By Lemma \ref{lem:states in Fk}, we have that $n=2\sum_{i\in \{1,\dots,k\}} p_i$. By the prime number theorem (see e.g. Newman~\cite{Newman}) we have that $\sum_{i\in \{1,\dots,k\}} p_i= \sum_{i\in \{1,\dots,k\}} o(k\log k) = o(k^2\log k)$.  

Let $f(x)=x^2\log x$ for $x> 1$. The function $f(x)$ is strictly monotone increasing and hence, has an inverse function. Let that function be $f^{-1}(y)$.  We have that $f^{-1}(y)\geq  \sqrt{\frac{y}{\log y}}$, for $y\geq 2$, because  \begin{eqnarray*}
f^{-1}(y)\geq \sqrt{\frac{y}{\log y}} & \Leftarrow & f(f^{-1}(y)) \geq f(\sqrt{\frac{y}{\log y}})\\
& \Leftarrow & y \geq (\sqrt{\frac{y}{\log y}})^2\log (\sqrt{\frac{y}{\log y}})\\
& \Leftarrow & y \geq \frac{y}{\log y} \log (\sqrt{\frac{y}{\log y}})\\
& \Leftarrow & y \geq \frac{y}{\log y} \log y\\
& \Leftarrow & y \geq y\\
\end{eqnarray*}

Here, the first $\Leftarrow$ follows by taking $f^{-1}$ on both sides. The function $f^{-1}$ is strictly monotone increasing, because $f(x)$ was. The fourth $\Leftarrow$ follows from $y\geq \sqrt{\frac{y}{\log y}}$ for $y\geq 2$ and $\log$ being monotone increasing.

Therefore, let $g(k)=2\sum_{i\in \{1,\dots,k\}} p_i$, then $g^{-1}(n)=\Omega(\sqrt{\frac{n}{\log n}})$.
By Lemma \ref{lem:bits needed for Fk}, we have that the period is $\prod_{i\in \{1,\dots,k\}} p_i$. Trivially we have that
\[
\prod_{i\in \{1,\dots,k\}} p_i\geq \prod_{i\in \{1,\dots,k\}} i =k!= 2^{\Omega ( k \log k)}
\]

We now insert $\Omega(\sqrt{\frac{n}{\log n}})$ in place of $k$ and get
\[
\prod_{i\in \{1,\dots,k\}} p_i= 2^{\Omega(\Omega(\sqrt{\frac{n}{\log n}}) \log (\Omega(\sqrt{\frac{n}{\log n}})))}=2^{\Omega(\sqrt{\frac{n}{\log n}} (\log n - \log \log n))}=2^{\Omega(\sqrt{n\log n})}
\]
The result follows.
\end{proof}

\section{Conclusion}

In the present paper we have considered properties of finite-horizon Markov decision processes and simple stochastic games. The $\epsilon$-optimal strategies considered in Section \ref{sec:counter} indicates the hardness of playing such games with a short horizon. The concept of period from Section \ref{sec:period} indicates the hardness of playing such games with a long horizon. Along with our lower bound from Section~\ref{sec:period} we conjecture the following: 
\begin{conjecture}
All FSSGs have an optimal strategy, which is an finite memory counter-based strategy, with period at most $2^n$. 
\end{conjecture}


\begin{thebibliography}{1}

\bibitem{Condon92}
A.~Condon.
\newblock The complexity of stochastic games.
\newblock {\em Information and Computation}, 96:203--224, 1992.

\bibitem{Everett}
H.~Everett.
\newblock Recursive games.
\newblock In H.~W. Kuhn and A.~W. Tucker, editors, {\em Contributions to the
  Theory of Games Vol. III}, volume~39 of {\em Annals of Mathematical Studies}.
  Princeton University Press, 1957.

\bibitem{FV97}
J.~Filar and K.~Vrieze.
\newblock {\em Competitive Markov Decision Process}, pages 16--22 (Chapter
  2.2).
\newblock Springer-Verlag, 1997.

\bibitem{Howard}
R.~A. Howard.
\newblock {\em Dynamic Programming and Markov Processes}.
\newblock M.I.T. Press, 1960.

\bibitem{IJM11}
R.~Ibsen-Jensen and P.~B. Miltersen.
\newblock Solving simple stochastic games with few coin toss positions.
\newblock {\em European Symposia on Algorithms, to appear}, 2012.

\bibitem{Newman}
D.~J. Newman.
\newblock Simple analytic proof of the prime number theorem.
\newblock {\em The American Mathematical Monthly}, 87(9):pp. 693--696, 1980.

\bibitem{Puterman}
M.~L. Puterman.
\newblock {\em Markov Decision Processes}, pages 74--118 (Chapter 4).
\newblock John Wiley \& Sons, Inc., 2008.

\end{thebibliography}
\end{document}